\documentclass[runningheads]{llncs}
\usepackage[vlined]{algorithm2e}

\newcommand{\Oh}[1]
    {\ensuremath{\mathcal{O}\!\left( {#1} \right)}}
\newcommand{\sfX}{\ensuremath{\mathsf{X}}}

\begin{document}

\title{Grammar-Based Compression\\in a Streaming Model}
\toctitle{Grammar-Based Compression in a Streaming Model}
\author{Travis Gagie\inst{1}\fnmsep\thanks
    {Funded by the Millennium Institute for Cell Dynamics and Biotechnology (ICDB),
    Grant ICM P05-001-F, Mideplan, Chile.}
    \and Pawe\l\ Gawrychowski\inst{2}}
\authorrunning{T. Gagie and P. Gawrychowski}
\tocauthor{Travis Gagie and Pawe\l\ Gawrychowski}
\institute{Department of Computer Science\\
    University of Chile\\
    \email{travis.gagie@gmail.com}
    \\\mbox{}\\
    \and Institute of Computer Science\\
    University of Wroc{\l}aw, Poland\\
    \email{gawry1@gmail.com}}
\maketitle

\begin{abstract}
We show that, given a string $s$ of length $n$, with constant memory and logarithmic passes over a constant number of streams we can build a context-free grammar that generates $s$ and only $s$ and whose size is within an $\Oh{\min \left( g \log g, \sqrt{n / \log n} \right)}$-factor of the minimum $g$.  This stands in contrast to our previous result that, with polylogarithmic memory and polylogarithmic passes over a single stream, we cannot build such a grammar whose size is within any polynomial of $g$.
\end{abstract}

\section{Introduction} \label{sec:intro}

In the past decade, the ever-increasing amount of data to be stored and manipulated has inspired intense interest in both grammar-based compression and streaming algorithms, resulting in many practical algorithms and upper and lower bounds for both problems.  Nevertheless, there has been relatively little study of grammar-based compression in a streaming model.  In a previous paper~\cite{Gag09} we proved limits on the quality of the compression we can achieve with polylogarithmic memory and polylogarithmic passes over a single stream.  In this paper we show how to achieve better compression with constant memory and logarithmic passes over a constant number of streams.

For grammar-based compression of a string $s$ of length $n$, we try to build a small context-free grammar (CFG) that generates $s$ and only $s$.  This is useful not only for compression but also for, e.g., indexing~\cite{CN09,Lif07} and speeding up dynamic programs~\cite{LMWZ09}.  (It is sometimes desirable for the CFG to be in Chomsky normal form (CNF), in which case it is also known as a straight-line program.)  We can measure our success in terms of universality~\cite{KY00}, empirical entropy~\cite{NR08} or the ratio between the size of our CFG and the size \(g = \Omega (\log n)\) of the smallest such grammar.  In this paper we consider the third and last measure.  Storer and Szymanski~\cite{SS82} showed that determining the size of the smallest grammar is NP-complete; Charikar et al.~\cite{CLLP+05} showed it cannot be approximated to within a small constant factor in polynomial time unless P\,$=$\,NP, and that even approximating it to within a factor of \(o (\log n / \log \log n)\) in polynomial time would require progress on a well-studied algebraic problem.  Charikar et al. and Rytter~\cite{Ryt03} independently gave $\Oh{\log (n / g)}$-approximation algorithms, both based on turning the LZ77~\cite{ZL77} parse of $s$ into a CFG and both initially presented at conference in 2002; Sakamoto~\cite{Sak05} then proposed another $\Oh{\log (n / g)}$-approximation algorithm, based on Re-Pair~\cite{LM00}.  Sakamoto, Kida and Shimozono~\cite{SKS04} gave a linear-time $\Oh{(\log g) (\log n)}$-approximation algorithm that uses $\Oh{g \log g}$ workspace, again based on LZ77; together with Maruyama, they~\cite{SMKS09} recently modified their algorithm to run in $\Oh{n \log^* n}$ time but achieve an $\Oh{(\log n) (\log^* n)}$ approximation ratio.

A few years before Charikar et al.'s and Rytter's papers sparked a surge of interest in grammar-based compression, a paper by Alon, Matias and Szegedy~\cite{AMS99} did the same for streaming algorithms.  We refer the reader to Babcock et al.~\cite{BBDMW02} and Muthukrishnan~\cite{Mut05} for a thorough introduction to streaming algorithms.  In this paper, however, we are most concerned with more powerful streaming models than these authors consider, ones that allow the use of multiple streams.  A number of recent papers have considered such models, beginning with Grohe and Schweikardt's~\cite{GS05} definition of \((r, s, t)\)-bounded Turing Machines, which use at most $r$ reversals over $t$ ``external-memory'' tapes and a total of at most $s$ space on ``internal-memory'' tapes to which they have unrestricted access.  While Munro and Paterson~\cite{MP80} proved tight bounds for sorting with one tape three decades ago, Grohe and Schweikardt proved the first tight bounds for sorting with multiple tapes.  Grohe, Hernich and Schweikardt~\cite{GHS09} proved lower bounds in this model for randomized algorithms with one-sided error, and Beame, Jayram and Rudra~\cite{BJR07} proved lower bounds for algorithms with two-sided error (renaming the model ``read/write streams'').  Beame and Huynh~\cite{BH08} revisited the problem considered by Alon, Matias and Szegedy, i.e., approximating frequency moments, and proved lower bounds for read/write stream algorithms.  Hernich and Schweikardt~\cite{HS08} related results for read/write stream algorithms to results in classical complexity theory, including results by Chen and Yap~\cite{CY91} on reversal complexity.  Hernich and Schweikardt's paper drew our attention to a theorem by Chen and Yap implying that, if a problem can be solved deterministically with read-only access to the input and logarithmic workspace then, in theory, it can be solved with constant memory and logarithmic passes (in either direction) over a constant number of read/write streams.  This theorem is the key to our main result in this paper.  Unfortunately, the constants involved in Chen and Yap's construction are enormous; we leave as future work finding a more practical proof of our results.

The study of compression in something like a streaming model goes back at least a decade, to work by Sheinwald, Lempel and Ziv~\cite{SLZ95} and De Agostino and Storer~\cite{DS96}.  As far as we know, however, our joint paper with Manzini~\cite{GM07} was the first to give nearly tight bounds in a standard streaming model.  In that paper we proved nearly matching bounds on the compression achievable with a constant amount of internal memory and one pass over the input, as well as upper and lower bounds for LZ77 with a sliding window whose size grows as a function of the number of characters encoded.  (Our upper bound for LZ77 used a theorem due to Kosaraju and Manzini~\cite{KM99} about quasi-distinct parsings, a subject recently revisited by Amir, Aumann, Levy and Roshko~\cite{AALR09}.)  Shortly thereafter, Albert, Mayordomo, Moser and Perifel~\cite{AMMP08} showed the compression achieved by LZ78~\cite{ZL78} is incomparable to that achievable by pushdown transducers; Mayordomo and Moser~\cite{MM09} then extended their result to show both kinds of compression are incomparable with that achievable by online algorithms with polylogarithmic memory.  (A somewhat similar subject, recognition of the context-sensitive Dyck languages in a streaming model, was recently broached by Magniez, Mathieu and Nayak~\cite{MMN09}, who gave a one-pass algorithm with one-sided error that uses polylogarithmic time per character and $\Oh{n^{1 / 2} \log n}$ space.)  In a recent paper with Ferragina and Manzini~\cite{FGM10} we demonstrated the practicality of streaming algorithms for compression in external memory.

In a recent paper~\cite{Gag09} we proved several lower bounds for compression algorithms that use a single stream, all based on an automata-theoretic lemma: suppose a machine implements a lossless compression algorithm using sequential accesses to a single tape that initially holds the input; then we can reconstruct any substring given, for every pass, the machine's configurations when it reaches and leaves the part of the tape that initially holds that substring, together with all the output it generates while over that part.  (We note similar arguments appear in computational complexity, where they are referred to as ``crossing sequences'', and in communication complexity.)  It follows that, if a streaming compression algorithm is restricted to using polylogarithmic memory and polylogarithmic passes over one stream, then there are periodic strings with polylogarithmic periods such that, even though the strings are very compressible as, e.g., CFGs, the algorithm must encode them using a linear number of bits; therefore, no such algorithm can approximate the smallest-grammar problem to within any polynomial of the minimum size.  Such arguments cannot prove lower bounds for algorithms with multiple streams, however, and we left open the question of whether extra streams allow us to achieve a polynomial approximation.  In this paper we use Chen and Yap's result to confirm they do: we show how, with logarithmic workspace, we can compute the LZ77 parse and turn that into a CFG in CNF while increasing the size by a factor of $\Oh{\min \left( g \log g, \sqrt{n / \log n} \right)}$ --- i.e., at most polynomially in $g$.  It follows that we can achieve that approximation ratio while using constant memory and logarithmic passes over a constant number of streams.

\section{LZ77 in a Streaming Model} \label{sec:LZ77}

Our starting point is the same as that of Charikar et al., Rytter and Sakamoto, Kida and Shimozono, but we pay even more attention to workspace than the last set of authors.  Specifically, we begin by considering the variant of LZ77 considered by Charikar et al. and Rytter, which does not allow self-referencing phrases but still produces a parse whose size is at most as large as that of the smallest grammar.  Each phrase in this parse is either a single character or a substring of the prefix of $s$ already parsed.  For example, the parse of ``{\sf how-much-wood-would-a-woodchuck-chuck-if-a-woodchuck-could-chuck-wood?}'' is ``{\sf h$|$o$|$w$|$-$|$m$|$u$|$c$|$h$|$-$|$w$|$o$|$
\linebreak
o$|$d$|$-wo$|$u$|$l$|$d-$|$a$|$-wood$|$ch$|$uc$|$k$|$-$|$chuck-$|$i$|$f$|$-a-woodchuck-c$|$ould-$|$chuck-$|$wood$|$?}''.

\begin{lemma}[Charikar et al., 2002; Rytter, 2002] \label{lem:C+R02}
The number of phrases in the LZ77 parse is a lower bound on the size of the smallest grammar.
\end{lemma}

\noindent As an aside, we note that Lemma~\ref{lem:C+R02} and results from our previous paper~\cite{Gag09} together imply we cannot compute the LZ77 parse with one stream when the product of the memory and passes is sublinear in $n$.

It is not difficult to show that this LZ77 parse --- like the original --- can be computed with read-only access to the input and logarithmic workspace.  Pseudocode for doing this appears as Algorithm~\ref{alg:parse}.  On the example above, this pseudocode produces ``{\sf how-much-wood(9,3)ul(13,2)a(9,5)(7,2)(6,2)k-(27,6)if(20,14)
\linebreak
(16,5)(27,6)(10,4)?}''.

\begin{lemma} \label{lem:lz77}
We can compute the LZ77 parse with logarithmic workspace.
\end{lemma}

\begin{proof}
The first phrase in the parse is the first letter in $s$; after outputting this, we always keep a pointer $t$ to the division between the prefix already parsed and the suffix yet to be parsed.  To compute each later phrase in turn, we check the length of the longest common prefix of \(s [i..t - 1]\) and \(s [t..n]\), for \(1 \leq i < t\); if the longest match has length 0 or 1, we output \(s [t]\); otherwise, we output the value of the minimal $i$ that maximizes the length of the longest common prefix, together with that length.  This takes a constant number of pointers into $s$ and a constant number of $\Oh{\log n}$-bit counters. \qed
\end{proof}

\begin{algorithm}[t]
\(t \leftarrow 1\)\;
\While{\(t \leq n\)}
    {\(\mathit{max\_match} \leftarrow 0\)\;
    \(\mathit{max\_length} \leftarrow 0\)\;
    \For{\(i \leftarrow 1 \ldots t - 1\)}
        {\(j \leftarrow 0\)\;
        \While{\(s [i + j] = s [t + j]\)}
            {\(j \leftarrow j + 1\)\;}
        \If{\(j > \mathit{max\_length}\)}
            {\(\mathit{max\_match} \leftarrow i\)\;
            \(\mathit{max\_length} \leftarrow j\)\;}}
    \eIf{\(\mathit{max\_length} \leq 1\)}
        {print \(s[t]\)\;
        \(t \leftarrow t + 1\)\;}
        {print \((\mathit{max\_match}, \mathit{max\_length})\)\;
        \(t \leftarrow t + max\_length\)\;}}
\BlankLine
\BlankLine

\caption{pseudocode for computing the LZ77 parse in logarithmic workspace}
\label{alg:parse}
\end{algorithm}

Combined with Chen and Yap's theorem below, Lemma~\ref{lem:lz77} implies that we can compute the LZ77 parse with constant workspace and logarithmic passes over a constant number of streams.

\begin{theorem}[Chen and Yap, 1991] \label{thm:CY91}
If a function can be computed with logarithmic workspace, then it can be computed with constant workspace and logarithmic passes over a constant number of streams.
\end{theorem}

As an aside, we note that Chen and Yap's theorem is actually much stronger than what we state here: they proved that, if \(f (n) = \Omega (\log n)\) is reversal-computable (see~\cite{CY91} or~\cite{HS08} for an explanation) and a problem can be solved deterministically in \(f (n)\) time, then it can be solved with constant workspace and $\Oh{f (n)}$ passes over a constant number of tapes.  Chen and Yap showed how a reversal-bounded Turing machine can simulate a space-bounded Turing machine by building a table of the possible configurations of the space-bounded machine.  Schweikardt~\cite{Sch07} pointed out that ``this is of no practical use, since the resulting algorithm produces huge intermediate results, but it is of major theoretical interest'' because it implies that a number of lower bounds are tight.  We leave as future work finding a more practical proof our our results.

In the next section we prove the following lemma, which is the most technical part of this paper.  By the size of a CFG, we mean the number of symbols on the righthand sides of the productions; notice this is at most a logarithmic factor less than the number of bits needed to express the CFG.

\begin{lemma} \label{lem:cfg}
With logarithmic workspace we can turn the LZ77 parse into a CFG whose size is within a $\Oh{\min \left( g \log g, \sqrt{n / \log n} \right)}$-factor of minimum.
\end{lemma}

\noindent Together with Lemma~\ref{lem:lz77} and Theorem~\ref{thm:CY91}, Lemma~\ref{lem:cfg} immediately implies our main result.

\begin{theorem}
With constant workspace and logarithmic passes over a constant number of streams, we can build a CFG generating $s$ and only $s$ whose size is within a $\Oh{\min \left( g \log g, \sqrt{n / \log n} \right)}$-factor of minimum.
\end{theorem}

\section{Logspace CFG Construction} \label{sec:grammar}

Unlike the LZ78 parse, the LZ77 parse cannot normally be viewed as a CFG, because the substring to which a phrase matches may begin or end in the middle of a preceding phrase.  We note this obstacle has been considered by other authors in other circumstances, e.g., by Navarro and Raffinot~\cite{NR04} for pattern matching.  Fortunately, we can remove this obstacle in logarithmic workspace, without increasing the number of phrases more than quadratically.  To do this, for each phrase for which we output \((i, \ell)\), we ensure \(s [i]\) is the first character in a phrase and \(s [i + \ell]\) is the last character in a phrase (by breaking phrases in two, if necessary).  For example, the parse
\begin{center}
``{\sf h$|$o$|$w$|$-$|$m$|$u$|$c$|$h$|$-$|$w$|$o$|$o$|$d$|$-wo$|$u$|$l$|$d-$|$a$|$-wood$|$ch$|$uc$|$k$|$-$|$chuck
\linebreak
-$|$i$|$f$|$-a-woodchuck-c$|$ould-$|$chuck-$|$wood$|$?}''
\end{center}
becomes
\begin{center}
``{\sf h$|$o$|$w$|$-$|$m$|$u$|$c$|$h$|$-$|$w$|$o$|$o$|$d$|$-w\ \raisebox{-.6ex}{\rule{.3ex}{2.5ex}}$_1$o$|$u$|$l$|$d\ \raisebox{-.6ex}{\rule{.3ex}{2.5ex}}$_2$-$|$a$|$-wood$|$c\ \raisebox{-.6ex}{\rule{.3ex}{2.5ex}}$_3$h$|$uc$|$k$|$-$|$c\ \raisebox{-.6ex}{\rule{.3ex}{2.5ex}}$_4^3$huck
\linebreak
-$|$i$|$f$|^2$-a-woodchuck-c$|^{1, 4}$ould-$|$chuck-$|$wood$|$?}'',
\end{center}
where the thick lines indicate new breaks and superscripts indicate which breaks cause the new ones (which are subscripted).  Notice the break ``{\sf a-woodchuck-c$|^{1, 4}$ould}'' causes both ``{\sf w\ \raisebox{-.6ex}{\rule{.3ex}{2.5ex}}$_1$ould}'' (matching ``{\sf ould}'') and ``{\sf a-woodchuck-c\ \raisebox{-.6ex}{\rule{.3ex}{2.5ex}}$_4^3$huck}'' (matching ``{\sf a-woodchuck-c}''); in turn, the latter new break causes ``{\sf woodc\ \raisebox{-.6ex}{\rule{.3ex}{2.5ex}}$_3$huck}'' (matching ``{\sf huck}''), which is why it has a superscript 3.

\begin{lemma} \label{lem:breaks}
Breaking the phrases takes at most logarithmic workspace and at most squares the number of phrases.  Afterwards, every phrase is either a single character or the concatenation of complete, consecutive, preceding phrases.
\end{lemma}

\begin{proof}
Since the phrases' start points are the partial sums of their lengths, we can compute them with logarithmic workspace; therefore, we can assume without loss of generality that the start points are stored with the phrases.  We start with the rightmost phrase and work left.  For each phrase's endpoints, we compute the corresponding position in the matching, preceding substring (notice that the position corresponding to one phrase's finish may not be the one corresponding to the start of the next phrase to the right) and insert a new break there, if there is not one already.  If we have inserted a new break, then we iterate, computing the position corresponding to the new break; eventually, we will reach a point where there is already a break, so the iteration will stop.  This process requires only a constant number of pointers, so we can perform it with logarithmic workspace.  Also, since each phrase is broken at most twice for each of the phrases that initially follow it in the parse, the final number of phrases is at most the square of the initial number.  By inspection, after the process is complete every phrase is the concatenation of complete, consecutive, preceding phrases. \qed
\end{proof}

Notice that, after we break the phrases as described above, we can view the parse as a CFG.  For example, the parse for our running example corresponds to
\[\begin{array}{l@{\hspace{5ex}}l@{\hspace{5ex}}l}
\begin{array}{rcl}
\sfX_0 & \rightarrow & \sfX_1 \ldots \sfX_{35}\\
\sfX_1 & \rightarrow & \mathsf{h}\\
& \vdots &\\
\sfX_{13} & \rightarrow & \mathsf{d}
\end{array} &
\begin{array}{rcl}
\sfX_{14} & \rightarrow & \sfX_9\ \sfX_{10}\\
\sfX_{15} & \rightarrow & \mathsf{o}\\
& \vdots &\\
\sfX_{31} & \rightarrow & \sfX_{19} \ldots \sfX_{27}
\end{array} &
\begin{array}{rcl}
& \vdots &\\
\sfX_{34} & \rightarrow & \sfX_{10} \ldots \sfX_{13}\\
\sfX_{35} & \rightarrow & \mbox{\sf ?}\\
&&
\end{array}
\end{array}\]
where $\sfX_0$ is the starting nonterminal.  Unfortunately, while the number of productions is polynomial in the number of phrases in the LZ77 parse, it is not clear the size is and, moreover, the grammar is not in CNF.  Since all the righthand sides of the productions are either terminals or sequences of consecutive nonterminals, we could put the grammar into CNF by squaring the number of nonterminals --- giving us an approximation ratio cubic in $g$.  This would still be enough for us to prove our main result but, fortunately, such a large increase is not necessary.

\begin{lemma} \label{lem:cnf}
Putting the CFG into CNF takes logarithmic workspace and increases the number of productions by at most a logarithmic factor.  Afterwards, the size of the grammar is proportional to the number of productions.
\end{lemma}

\begin{proof}
We build a forest of complete binary trees whose leaves are the nonterminals: if we consider the trees in order by size, the nonterminals appear in order from the leftmost leaf of the first tree to the rightmost leaf of the last tree; each tree is as large as possible, given the number of nonterminals remaining after we build the trees to its left.  Notice there are $\Oh{\log g}$ such trees, of total size at most $\Oh{g^2}$.  We then assign a new nonterminal to each internal node and output a production which takes that nonterminal to its children.  This takes logarithmic workspace and increases the number of productions by a constant factor.

Notice any sequence of consecutive nonterminals that spans at least two trees, can be written as the concatenation of two consecutive sequences, one of which ends with the rightmost leaf in one tree and the other of which starts with the leftmost leaf in the next tree.  Consider a sequence ending with the rightmost leaf in a tree; dealing with one that starts with a leftmost leaf is symmetric.  If the sequence completely contains that tree, we can write a binary production that splits the sequence into the prefix in the preceding trees, which is the expansion of a new nonterminal, and the leaves in that tree, which are the expansion of its root.  We need do this $\Oh{\log g}$ times before the remaining subsequence is contained within a single tree.  After that, we repeatedly produce new binary productions that split the subsequence into prefixes, again the expansions of new nonterminals, and suffixes, the expansions of roots of the largest possible complete subtree.  Since the size of the largest possible complete subtree shrinks by a factor of two at each step (or, equivalently, the height of its root decreases by 1), we need repeat $\Oh{\log g}$ times.  Again, this takes logarithmic workspace (we will give more details in the full version of this paper).

In summary, we may replace each production with $\Oh{\log g}$ new, binary productions.  Since the productions are binary, the number of symbols on the righthand sides is linear in the number of productions themselves. \qed
\end{proof}

Lemma~\ref{lem:cnf} is our most detailed result, and the diagram below showing the construction with our running example is also somewhat detailed.  On the left are modifications of the original productions, now made binary; in the middle are productions for the internal nodes of the binary trees; and on the right are productions breaking down the consecutive subsequences that appear on the righthand sides of the productions in the left column, until the subsequences are single, original nonterminals or nonterminals for nodes in the binary trees (i.e., those on the lefthand sides of the productions in the middle column).

\[\begin{array}{l@{\hspace{5ex}}l@{\hspace{5ex}}l}
\begin{array}{rcl}
\sfX_0 & \rightarrow & \sfX_{1, 32}\ \sfX_{33, 35}\\
\sfX_1 & \rightarrow & \mathsf{h}\\
& \vdots &\\
\sfX_{13} & \rightarrow & \mathsf{d}\\
\sfX_{14} & \rightarrow & \sfX_9\ \sfX_{10}\\
\sfX_{15} & \rightarrow & \mathsf{o}\\
& \vdots &\\
\sfX_{31} & \rightarrow & \sfX_{19, 24}\ \sfX_{25, 27}\\
& \vdots &\\
\sfX_{34} & \rightarrow & \sfX_{10, 12}\ \sfX_{13}\\
\sfX_{35} & \rightarrow & \mbox{\sf ?}
\end{array} &
\begin{array}{rcl}
\sfX_{1, 32} & \rightarrow & \sfX_{1, 16}\ \sfX_{17, 32}\\
\sfX_{1, 16} & \rightarrow & \sfX_{1, 8}\ \sfX_{9, 16}\\
\sfX_{1, 8} & \rightarrow & \sfX_{1, 4}\ \sfX_{5, 8}\\
& \vdots &\\
\sfX_{17, 32} & \rightarrow & \sfX_{17, 24}\ \sfX_{18, 32}\\
\sfX_{17, 24} & \rightarrow & \sfX_{17, 20}\ \sfX_{21, 24}\\
& \vdots &\\
\sfX_{29, 32} & \rightarrow & \sfX_{29, 30}\ \sfX_{31, 32}\\
\sfX_{33, 35} & \rightarrow & \sfX_{33, 34}\ \sfX_{35}\\
\sfX_{1, 2} & \rightarrow & \sfX_1\ \sfX_2\\
& \vdots &\\
\sfX_{33, 34} & \rightarrow & \sfX_{33}\ \sfX_{34}
\end{array} &
\begin{array}{rcl}
\sfX_{19, 24} & \rightarrow & \sfX_{19, 20}\ \sfX_{21, 24}\\
\sfX_{25, 27} & \rightarrow & \sfX_{25, 26}\ \sfX_{27}\\
& \vdots &\\
\sfX_{10, 12} & \rightarrow & \sfX_{10}\ \sfX_{11, 12}
\end{array}
\end{array}\]

Combined with Lemma~\ref{lem:C+R02}, Lemmas~\ref{lem:breaks} and~\ref{lem:cnf} imply that with logarithmic workspace we can build a CFG in CNF whose size is $\Oh{g^2 \log g}$.  We can use a similar approach with binary trees to build a CFG in CNF of size $\Oh{n}$ that generates $s$ and only $s$, still using logarithmic workspace.  If we combine all non-terminals that have the same expansion, which also takes logarithmic workspace, then this becomes Kieffer, Yang, Nelson and Cosman's~\cite{KYNC00} {\sc Bisection} algorithm, which gives an $\Oh{\sqrt{n / \log n}}$-approximation~\cite{CLLP+05}.  By taking the smaller of these two CFGs we achieve an $\Oh{\min \left( g \log g, \sqrt{n / \log n} \right)}$-approximation.  Therefore, as we claimed in Lemma~\ref{lem:cfg}, with logarithmic workspace we can turn the LZ77 parse into a CFG whose size is within a $\Oh{\min \left( g \log g, \sqrt{n / \log n} \right)}$-factor of minimum.

\section{Recent Work} \label{sec:recent}

We recently improved the bound on the approximation ratio in Lemma~\ref{lem:cfg} from $\Oh{\min \left( g \log g, \sqrt{n / \log n} \right)}$ to $\Oh{\min \left( g, \sqrt{n / \log n} \right)}$.  The key observation is that, by the definition of the LZ77 parse, the first occurrence of any substring must touch or cross a break between phrases.  Consider any phrase in the parse obtained by applying Lemma~\ref{lem:breaks} to the LZ77 parse.  By the observation above, that phrase can be written as the concatenation of some consecutive new phrases (all contained within one old phrase and ending at that old phrase's right end), some consecutive old phrases, and some more consecutive new phrases (all contained within one old phrase and starting at the old phrase's left end).  Since there are $\Oh{g}$ old phrases, there are $\Oh{g^2}$ sequences of consecutive old phrases; since there are $\Oh{g^2}$ new phrases, there are $\Oh{g^2}$ sequences of consecutive new phrases that are contained in one old phrase and either start at that old phrase's right end or end at that old phrase's left end.

While working on the improvement above, we realized how to improve the bound further, to $\Oh{\min \left( g, 4^{\sqrt{\log n}} \right)}$.  To do this, we choose a value $b$ between 2 and $n$ and, for \(0 \leq i \leq \log_b n\), we associate a nonterminal to each of the $b$ blocks of \(\lceil n / b^i \rceil\) characters to the left and right of each break; we thus start building the grammar with $\Oh{b g \log_b n}$ nonterminals.  We then add $\Oh{b g \log_b n}$ binary productions such that any sequence of nonterminals associated with a consecutive sequence of blocks, can be derived from $\Oh{1}$ nonterminals.  Notice any substring is the concatenation of 0 or 1 partial blocks, some number of full blocks to the left of a break, some number of blocks to the right of a break, and 0 or 1 more partial blocks.  We now add more binary productions as follows: we start with $s$ (the only block of length \(\left\lceil n / b^0 \right\rceil = n\)); find the first break it touches or crosses (in this case it is the start of $s$); consider $s$ as the concatenation of blocks of size \(\left\lceil n / b^1 \right\rceil\) (in this case only the rightmost block can be partial); associate nonterminals to the partial blocks (if they exist); add $\Oh{1}$ productions to take the symbol associated to $s$ (in this case, the start symbol) to the sequence of nonterminals associated with the smaller blocks in order from left to right; and recurse on each of the smaller blocks.  To guarantee each smaller block touches or crosses a break, we work on the first occurrence in $s$ of the substring contained in that block.  We stop recursing when the block size is 1, and add $\Oh{b g}$ productions taking those blocks' nonterminals to the appropriate characters.

Analysis shows that the number of productions we add during the recursion is proportional to the number of blocks involved, either full or partial.  Since the number of distinct full blocks in any level of recursion is $\Oh{b g}$ and the number of partial blocks is at most twice the number of blocks (full or partial) in the previous level of recursion, the number of productions we add during the recursion is $\Oh{2^{\log_b n} b g}$.  Therefore, the grammar has size $\Oh{2^{\log_b n} b g}$; when \(b = 2^{\sqrt{\log n}}\), this is $\Oh{4^{\sqrt{\log n}} g}$.  The first of the two key observations that let us build the grammar in logarithmic workspace, is that we can store the index of a block (full or partial) with respect to the associated break, in $\Oh{\sqrt{\log n}}$ bits; therefore, we can store $\Oh{1}$ indices for each of the $\Oh{\sqrt{\log n}}$ levels of recursion, in a total of $\Oh{\log n}$ bits.  The second key observation is that, given the indices of the block we are working on in each level of recursion, with respect to the appropriate break, we can compute the start point and end point of the block we are currently working on in the deepest level of recursion.  We will give details of these two improvements in the full version of this paper.

While working on this second improvement, we realized that we can use the same ideas to build a compressed representation that allows efficient random access.  We refer the reader to the recent papers by Kreft and Navarro~\cite{KN10} and Bille, Landau and Weimann~\cite{BLW10} for background on this problem.  Suppose that, for each of the $\Oh{2^{\sqrt{\log n}} g \sqrt{\log n}}$ full blocks described above, we store a pointer to the first occurrence in $s$ of the substring in that block, as well as a pointer to the first break that first occurrence touches or crosses.  Notice this takes a total of $\Oh{2^{\sqrt{\log n}} g (\log n)^{3 / 2}}$ bits.  Then, given a block's index and an offset in that block, in $\Oh{1}$ time we can compute a smaller block's index and offset in that smaller block, such that the characters in those two positions are equal; if the larger block has size 1, in $\Oh{1}$ time we can return the character.  Since $s$ itself is a block, an offset in it is just a character's position, and there are $\Oh{\sqrt{\log n}}$ levels of recursion, it follows that we can access any character in $\Oh{\sqrt{\log n}}$ time.  Further analysis shows that it takes $\Oh{\sqrt{\log n} + \ell}$ time to access a substring of length $\ell$.  Of course, for any positive constant $\epsilon$, if we are willing to use $\Oh{n^{\epsilon} g}$ bits of space, then we can access any character in constant time.  If we make the data structure slightly larger and more complicated --- e.g., storing searchable partial sums at the block boundaries --- then, at least for strings over fairly small alphabets, we can also support fast rank and select queries.

This implementation makes it easy to see the data structure's relation to LZ77 and grammar-based compression.  We can use a simpler implementation, however, and use LZ77 only in the analysis.  Suppose that, for \(0 \leq i \leq \log_b n\), we break $s$ into consecutive blocks of length \(\lceil n / b^i \rceil\) (the last block may be shorter), always starting from the first character of $s$.  For each block, we store a pointer to the first occurrence in $s$ of that block's substring.  Given a block's index and an offset in that block, in $\Oh{1}$ time we can again compute a smaller block's index and offset in that smaller block, such that the characters in those two positions are equal: the new block's index is the sum of pointer and the old offset, divided by the new block length and rounded down; the new offset is the sum of the pointer and the old offset, modulo the new block length.  We can discard any block that cannot be visited during a query, so this data structure takes at most a constant factor more space than the one described above.  Indeed, this data structure seems likely to be smaller in practice, because blocks of the same size can overlap in the previous data structure but cannot in this one.  We plan to implement this data structure and report the results in a future paper.

\bibliographystyle{splncs}
\bibliography{lata10}

\begin{thebibliography}{10}

\bibitem{Gag09}
Gagie, T.:
\newblock On the value of multiple read/write streams for data compression.
\newblock In: Proceedings of the Symposium on Combinatorial Pattern Matching.
  (2009)  68--77

\bibitem{CN09}
Claude, F., Navarro, G.:
\newblock Self-indexed text compression using straight-line programs.
\newblock In: Proceedings of the Symposium on Mathematical Foundations of
  Computer Science. (2009)  235--246

\bibitem{Lif07}
Lifshits, Y.:
\newblock Processing compressed texts: A tractability border.
\newblock In: Proceedings of the Symposium on Combinatorial Pattern Matching.
  (2007)  228--240

\bibitem{LMWZ09}
Lifshits, Y., Mozes, S., Weimann, O., {Ziv-Ukelson}, M.:
\newblock Speeding up {HMM} decoding and training by exploiting sequence
  repetitions.
\newblock Algorithmica \textbf{54}(3) (2009)  379--399

\bibitem{KY00}
Kieffer, J.C., Yang, E.:
\newblock Grammar-based codes: A new class of universal lossless source codes.
\newblock IEEE Transactions on Information Theory \textbf{46}(3) (2000)
  737--754

\bibitem{NR08}
Navarro, G., Russo, L.M.S.:
\newblock Re-pair achieves high-order entropy.
\newblock In: Proceedings of the Data Compression Conference. (2008)  537

\bibitem{SS82}
Storer, J.A., Szymanski, T.G.:
\newblock Data compression via textual substitution.
\newblock Journal of the ACM \textbf{29}(4) (1982)  928--951

\bibitem{CLLP+05}
Charikar, M., Lehman, E., Liu, D., Panigrahy, R., Prabhakaran, M., Sahai, A.,
  Shelat, A.:
\newblock The smallest grammar problem.
\newblock IEEE Transactions on Information Theory \textbf{51}(7) (2005)
  2554--2576

\bibitem{Ryt03}
Rytter, W.:
\newblock Application of {Lempel-Ziv} factorization to the approximation of
  grammar-based compression.
\newblock Theoretical Computer Science \textbf{302}(1--3) (2003)  211--222

\bibitem{ZL77}
Ziv, J., Lempel, A.:
\newblock A universal algorithm for sequential data compression.
\newblock IEEE Transactions on Information Theory \textbf{23}(3) (1977)
  337--343

\bibitem{Sak05}
Sakamoto, H.:
\newblock A fully linear-time approximation algorithm for grammar-based
  compression.
\newblock Journal of Discrete Algorithms \textbf{3}(2--4) (2005)  416--430

\bibitem{LM00}
Larsson, N.J., Moffat, A.:
\newblock Offline dictionary-based compression.
\newblock Proceedings of the IEEE \textbf{88}(11) (2000)  1722--1732

\bibitem{SKS04}
Sakamoto, H., Kida, T., Shimozono, S.:
\newblock A space-saving linear-time algorithm for grammar-based compression.
\newblock In: Proceedings of the Symposium on String Processing and Information
  Retrieval. (2004)  218--229

\bibitem{SMKS09}
Sakamoto, H., Maruyama, S., Kida, T., Shimozono, S.:
\newblock A space-saving approximation algorithm for grammar-based compression.
\newblock IEICE Transactions \textbf{92-D}(2) (2009)  158--165

\bibitem{AMS99}
Alon, N., Matias, Y., Szegedy, M.:
\newblock The space complexity of approximating the frequency moments.
\newblock Journal of Computer and System Sciences \textbf{58}(1) (1999)
  137--147

\bibitem{BBDMW02}
Babcock, B., Babu, S., Datar, M., Motwani, R., Widom, J.:
\newblock Models and issues in data stream systems.
\newblock In: Proceedings of the Symposium on Database Systems. (2002)  1--16

\bibitem{Mut05}
Muthukrishnan, S.:
\newblock Data Streams: Algorithms and Applications. Volume 1(2) of Foundations
  and Trends in Theoretical Computer Science.
\newblock now Publishers (2005)

\bibitem{GS05}
Grohe, M., Schweikardt, N.:
\newblock Lower bounds for sorting with few random accesses to external memory.
\newblock In: Proceedings of the Symposium on Database Systems. (2005)
  238--249

\bibitem{MP80}
Munro, J.I., Paterson, M.:
\newblock Selection and sorting with limited storage.
\newblock Theoretical Computer Science \textbf{12} (1980)  315--323

\bibitem{GHS09}
Grohe, M., Hernich, A., Schweikardt, N.:
\newblock Lower bounds for processing data with few random accesses to external
  memory.
\newblock Journal of the ACM \textbf{56}(3) (2009)  1--58

\bibitem{BJR07}
Beame, P., Jayram, T.S., Rudra, A.:
\newblock Lower bounds for randomized read/write stream algorithms.
\newblock In: Proceedings of the Symposium on Theory of Computing. (2007)
  689--698

\bibitem{BH08}
Beame, P., Huynh, T.:
\newblock On the value of multiple read/write streams for approximating
  frequency moments.
\newblock In: Proceedings of the Symposium on Foundations of Computer Science.
  (2008)  499--508

\bibitem{HS08}
Hernich, A., Schweikardt, N.:
\newblock Reversal complexity revisited.
\newblock Theoretical Computer Science \textbf{401}(1--3) (2008)  191--205

\bibitem{CY91}
Chen, J., Yap, C.:
\newblock Reversal complexity.
\newblock SIAM Journal on Computing \textbf{20}(4) (1991)  622--638

\bibitem{SLZ95}
Sheinwald, D., Lempel, A., Ziv, J.:
\newblock On encoding and decoding with two-way head machines.
\newblock Information and Computation \textbf{116}(1) (1995)  128--133

\bibitem{DS96}
{De Agostino}, S., Storer, J.A.:
\newblock On-line versus off-line computation in dynamic text compression.
\newblock Information Processing Letters \textbf{59}(3) (1996)  169--174

\bibitem{GM07}
Gagie, T., Manzini, G.:
\newblock Space-conscious compression.
\newblock In: Proceedings of the Symposium on Mathematical Foundations of
  Computer Science. (2007)  206--217

\bibitem{KM99}
Kosaraju, S.R., Manzini, G.:
\newblock Compression of low entropy strings with {Lempel-Ziv} algorithms.
\newblock SIAM Journal on Computing \textbf{29}(3) (1999)  893--911

\bibitem{AALR09}
Amir, A., Aumann, Y., Levy, A., Roshko, Y.:
\newblock Quasi-distinct parsing and optimal compression methods.
\newblock In: Proceedings of the Symposium on Combinatorial Pattern Matching.
  (2009)  12--25

\bibitem{AMMP08}
Albert, P., Mayordomo, E., Moser, P., Perifel, S.:
\newblock Pushdown compression.
\newblock In: Proceedings of the Symposium on Theoretical Aspects of Computer
  Science. (2008)  39--48

\bibitem{ZL78}
Ziv, J., Lempel, A.:
\newblock Compression of individual sequences via variable-rate coding.
\newblock IEEE Transactions on Information Theory \textbf{24}(5) (1978)
  530--536

\bibitem{MM09}
Mayordomo, E., Moser, P.:
\newblock Polylog space compression is incomparable with {Lempel-Ziv} and
  pushdown compression.
\newblock In: Proceedings of the Conference on Current Trends in Theory and
  Practice of Informatics. (2009)  633--644

\bibitem{MMN09}
Magniez, F., Mathieu, C., Nayak, A.:
\newblock Recognizing well-parenthesized expressions in the streaming model.
\newblock Technical Report TR09-119, Electronic Colloquium on Computational
  Complexity (2009)

\bibitem{FGM10}
Ferragina, P., Gagie, T., Manzini, G.:
\newblock Lightweight data indexing and compression in external memory.
\newblock In: Proceedings of the Latin American Theoretical Informatics
  Symposium. (2010) To appear.

\bibitem{Sch07}
Schweikardt, N.:
\newblock Machine models and lower bounds for query processing.
\newblock In: Proceedings of the Symposium on Principles of Database Systems.
  (2007)  41--52

\bibitem{NR04}
Navarro, G., Raffinot, M.:
\newblock Practical and flexible pattern matching over {Ziv-Lempel} compressed
  text.
\newblock Journal of Discrete Algorithms \textbf{2}(3) (2004)  347--371

\bibitem{KYNC00}
Kieffer, J.C., Yang, E., Nelson, G.J., Cosman, P.C.:
\newblock Universal lossless compression via multilevel pattern matching.
\newblock IEEE Transactions on Information Theory \textbf{46}(4) (2000)
  1227--1245

\bibitem{KN10}
Kreft, S., Navarro, G.:
\newblock {LZ77}-like compression with fast random access.
\newblock In: Proceedings of the Data Compression Conference. (2010) To appear.

\bibitem{BLW10}
Bille, P., Landau, G., Weimann, O.:
\newblock Random access to grammar compressed strings.
\newblock {\tt http://arxiv.org/abs/1001.1565} (2010)

\end{thebibliography}

\end{document}